\documentclass[11pt,a4paper]{article}
\usepackage{geometry}
\usepackage{amsthm}
\usepackage{amsfonts,colortbl}
\usepackage{amsmath,bm,hyperref,multirow}
\usepackage{graphicx,subcaption}
\usepackage{tikz,hyperref}
\usetikzlibrary{fit,calc,shapes,backgrounds,matrix,arrows}

\newtheorem{theorem}{Theorem}
\newtheorem{proposition}[theorem]{Proposition}
\newtheorem{example}[theorem]{Example}
\newtheorem{definition}[theorem]{Definition}

\newtheorem{remark}[theorem]{Remark}

\newtheorem{corollary}[theorem]{Corollary}

\newtheorem{property}[theorem]{Property}

\newcommand{\argmax}{\mathrm{argmax}}

\newcommand{\TN}{\mathrm{2TN}}
\newcommand{\lwTN}{\ensuremath{(l,w)-\mathrm{2TN}}}
\newcommand{\lwMMN}{\ensuremath{(l,w)-\mathrm{MMN}}}
\newcommand{\wlMMN}{\ensuremath{(w,l)-\mathrm{MMN}}}

\newcommand{\Rel}{\mathrm{Rel}}
\newcommand{\ApRel}{\mathrm{ApRel}} 

\newcommand{\Nss}{\bm{\mathrm{G}}}

\newcommand{\Hs}{\ensuremath{\bm{\mathrm{H}}}}

\newcommand{\PoS}{\bm{\mathrm{PoS}}}
\newcommand{\SoP}{\bm{\mathrm{SoP}}}

\title{\textbf{Generalized Convexity Properties and Shape Based Approximation in Networks Reliability}}
\author{Gabriela Cristescu$^1$ \and Vlad-Florin Dr\u{a}goi$^1$ \and Sorin-Hora\c{t}iu Hoar\u{a}$^{1,2}$}
\date{\scriptsize $^1$ Faculty of Exact Sciences, Aurel Vlaicu University of Arad, Romania\\
$^2$ Department of Computers and Information Technology, Polytechnic University of Timisoara}

\begin{document}

\thispagestyle{plain}

\maketitle
\begin{abstract}
Some properties of generalized convexity for sets and for functions are identified in case of the reliability polynomials of two dual minimal networks. A method of approximating the reliability polynomials of two dual minimal network is developed based on their mutual complementarity properties. The approximating objects are from the class of quadratic spline functions, constructed based both on interpolation conditions and on shape knowledge. It is proved that the approximant objects preserve the shape properties of the exact reliability polynomials. Numerical examples and simulations show the performance of the algorithm, both in terms of low complexity, small error and shape preserving. Possibilities of increasing the accuracy of approximation are discussed. 
\end{abstract}

\section{Introduction}

From the earliest days of network reliability, researchers have tried to develop algorithms to efficiently compute the reliability of graphs/networks. In a recent survey paper \cite{BCCGM2020}, Brown et \textit{al.} have rediscovered the way of this research domain, by putting into light some of the theoretical advances made in the past, as well as the new directions. This scientific adventure started with the work of Moore and Shannon \cite{1956_MS_1,1956_MS_2} and von Neumann \cite{1952_vN}, when the first foundations of the field were settled. In the late seventies \cite{1979_V}, Valiant demonstrated that the main computational problem, i.e., to compute the reliability polynomial of a two-terminal network, is \#P-complete. Hence, when the graph parameters are growing significantly, one has to find alternative methods for estimating reliability, such as i) applying simplifications in order to reduce the computations as much as possible so that the algorithm becomes practically effective, ii) bounding the reliability polynomial (using combinatorial methods and/or structural properties) iii) approximating the reliability polynomial such that the error of approximation is bounded by a relatively small quantity. 

The less investigated topics related to network reliability are the analytical properties such as shape properties of the reliability polynomials, including convexity, the number of real roots and their density, etc. In order to compensate the complexity problems of computing the coefficients of the reliability polynomial of a two-terminal network, the authors in \cite{CD2020,CD2021} proposed to approximate the polynomials using structural properties of the networks. In particular, duality is a characteristic that induces complementary properties on the coefficients, which are considered in the approximations. In \cite{DJ2020} Hermite interpolation is used for hammock networks based on previous results on the shape \cite{DJ2019}. Cubic splines are proposed in \cite{CD2020,CD2021}, and are suitable for any two-terminal networks. In \cite{CD2021} two methods of producing cubic splines are compared, Lagrange-type interpolation procedures and Bernstein approximation operator, emphasizing the accuracy of the methods. There are several advantages for taking duality into account in the context of approximations, such as
\begin{itemize}
    \item Computing the first non-trivial (different from zero) coefficient of the reliability polynomial of a network enables one to directly obtain the value of the last non-trivial (different from the binomial coefficient) coefficient of the reliability polynomial of the dual network;
    \item Adjusting approximated coefficients of a network can be done more efficiently when duality is considered, as more information is taken into account;
    \item The error of simultaneous approximation of two dual networks can be more accurately estimated when compared to a single network approximation.  
\end{itemize}

\paragraph{Our contribution}
The mutual behaviour of two dual networks from the point of view of their reliability provides us with additional information, which is used as input in the construction of the algorithms from \cite{CD2020,CD2021}. In this paper, we refine the approximation technique by considering the shape properties of the reliability polynomials of two dual two-terminal networks. A profound research on the shape of the reliability polynomials of two dual two-terminal networks is presented, starting from their complementarity properties \cite{CD2020,CD2021}. Their mutual behaviour referring to high order convexity properties, tangents properties, inflection points, etc. are emphasized. We construct approximation operators that preserve as many as possible shape properties. The use of the quadratic spline functions allows us to keep control on the approximation process from the point of view of shape preserving.

\paragraph{Outline of the article}
Allover the paper $\mathbb{N}$ denotes the set of natural numbers and $\mathbb{R}$ means the set of real numbers. A two terminal network is referred as 2TN and a matchstick minimal network is referred as MMN all over the paper.\\
In Section \ref{sec:net-reliab}, we describe the types of networks implied in our research and introduce the main definitions and properties regarding their reliability. Section \ref{sec:Shape-Prop} contains results referring to high order convexity properties of the reliability polynomial of a MMN and the manner, in which these properties are transferred to the dual network. Some extremum properties of the coefficients function are discussed. An algorithm for simultaneous approximation of the reliability polynomials of two dual network is described in Section \ref{sec:Approx-algo} based on quadratic spline functions. The approximant functions are constructed to preserve, as mush as possible, the shape properties of the reliability polynomials of MMNs. In Section \ref{sec:5}, we simulate the approximation technique using the new algorithm. The shape properties of the approximant objects are emphasized, proving the performance of the algorithm from shape preserving point of view. We conclude the article in Section \ref{sec:conclusion}. 

\section{Preliminaries on network reliability}\label{sec:net-reliab}
\subsection{Matchstick minimal two-terminal networks}\label{subsec:MMNs}
Any network $\Nss$ made of $n$ identical devices, that has two distinguished terminals: a source $S$, and a terminus $T$ is called two terminal network (denoted by $\TN$ in the sequel). $\Nss$ can be characterized at least by three parameters: {\it width} ($w$), {\it length} ($l$), and {\it size} $(n)$ - $w$ is the size of a "minimal cut" separating $S$ from $T$, $l$ is the size of a "minimal path" from $S$ to $T.$ In general, we have $n\ge wl$ (see \cite{1956_MS_1}). Any $\TN$ $\Nss$ of width $w$ and length $l$ will be called a $\lwTN$, or $\TN$ of type $(l,w)$. All functions, as reliability polynomial, its coefficients function and various approximations, related to a $\lwTN$ will be denoted by a character having the index $(l,w).$ Any $\TN$ satisfying $n=wl$ is minimal (see \cite{1956_MS_1}), the members of the family of minimal $\TN$s being denoted by MMN.
 
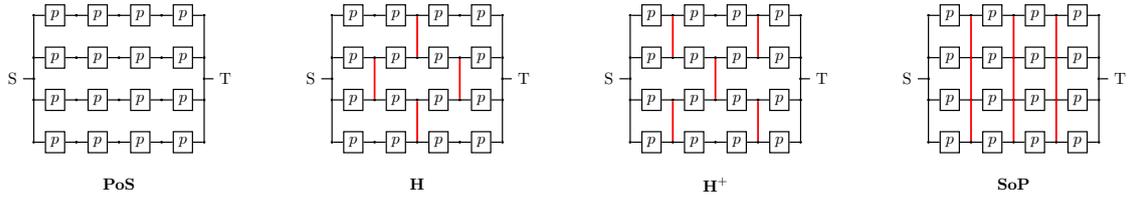
\begin{figure}[!ht]
\begin{center}
\resizebox{\textwidth}{!}{
\begin{tikzpicture}[thick]
\tikzset{
phase/.style = {draw,fill,shape=circle,minimum size=1pt,inner sep=0pt},
dev/.style={draw,shape=rectangle,minimum width =0.05 cm},
}

\node (S5) at (-8-7,-5) {S};
\node (T5) at (-3-7,-5) {T};

\node[phase] (y1) at (-7.5-7,-5){}; 
\node[phase] (y2) at (-3.5-7,-5){};

\draw[thick] (S5) -- (y1);
\draw[thick] (T5) -- (y2);
\foreach \j in {1,2,3,4}{
	\foreach \k in {1,2,3,4,5}{
		\node[phase] (x\k\j) at (\k-8.5-7,-7.5+\j){};
    }
}

\draw[thick] (x11) -- (x14);
\draw[thick] (x51) -- (x54);
	\foreach \j in {1,2,3,4} 
    {
    	\foreach \k in {1,2,3,4}{
			\node[dev] (w\k\j) at (\k-8-7,-7.5+\j) {$p$};
		}
    }
 \foreach \i in {1,2,3,4}{   
\draw[thick] (x1\i) -- (w1\i) --(w2\i) -- (w3\i) --(w4\i) -- (x5\i);
}
 \node(c) at (-5.5-7,-7.5) {$\PoS$};
 
\node (S5) at (-8,-5) {S};
\node (T5) at (-3,-5) {T};

\node[phase] (y1) at (-7.5,-5){}; 
\node[phase] (y2) at (-3.5,-5){};

\draw[thick] (S5) -- (y1);
\draw[thick] (T5) -- (y2);
\foreach \j in {1,2,3,4}{
	\foreach \k in {1,2,3,4,5}{
		\node[phase] (x\k\j) at (\k-8.5,-7.5+\j){};
    }
}

\draw[thick] (x11) -- (x14);
\draw[thick] (x51) -- (x54);
	\foreach \j in {1,2,3,4} 
    {
    	\foreach \k in {1,2,3,4}{
			\node[dev] (w\k\j) at (\k-8,-7.5+\j) {$p$};
		}
    }
 \foreach \i in {1,2,3,4}{   
\draw[thick] (x1\i) -- (w1\i) --(w2\i) -- (w3\i) --(w4\i) -- (x5\i);
}
    
\draw[very thick,color=red] (x31) -- (x32);
\draw[very thick,color=red] (x33) -- (x34);

\draw[very thick,color=red] (x22) -- (x23);
\draw[very thick,color=red] (x42) -- (x43);

\node(c) at (-5.5,-7.5) {$\Hs$};


\node (S5) at (-1,-5) {S};
\node (T5) at (4,-5) {T};
\node[phase] (y1) at (-0.5,-5){}; 
\node[phase] (y2) at (3.5,-5){};

\draw[thick] (S5) -- (y1);
\draw[thick] (T5) -- (y2);

\foreach \j in {1,2,3,4}{
	\foreach \k in {1,2,3,4,5}{
		\node[phase] (x\k\j) at (\k-1.5,-7.5+\j){};
    }
}

\draw[thick] (x11) -- (x14);
\draw[thick] (x51) -- (x54);
	\foreach \j in {1,2,3,4} 
    {
    	\foreach \k in {1,2,3,4}{
			\node[dev] (w\k\j) at (\k-1,-7.5+\j) {$p$};
		}
    }    
   \foreach \i in {1,2,3,4}{   
\draw[thick] (x1\i) -- (w1\i) --(w2\i) -- (w3\i) --(w4\i) -- (x5\i);
}  

\draw[very thick,color=red] (x21) -- (x22);
\draw[very thick,color=red] (x23) -- (x24);

\draw[very thick,color=red] (x43) -- (x44);

\draw[very thick,color=red] (x32) -- (x33);
\draw[very thick,color=red] (x41) -- (x42);

\node(c) at (1.5,-7.5) {$\Hs^{+}$};

\node (S5) at (-8+14,-5) {S};
\node (T5) at (-3+14,-5) {T};

\node[phase] (y1) at (-7.5+14,-5){}; 
\node[phase] (y2) at (-3.5+14,-5){};

\draw[thick] (S5) -- (y1);
\draw[thick] (T5) -- (y2);
\foreach \j in {1,2,3,4}{
	\foreach \k in {1,2,3,4,5}{
		\node[phase] (x\k\j) at (\k-8.5+14,-7.5+\j){};
    }
}

\draw[thick] (x11) -- (x14);
\draw[thick] (x51) -- (x54);
	\foreach \j in {1,2,3,4} 
    {
    	\foreach \k in {1,2,3,4}{
			\node[dev] (w\k\j) at (\k-8+14,-7.5+\j) {$p$};
		}
    }
 \foreach \i in {1,2,3,4}{   
\draw[thick] (x1\i) -- (w1\i) --(w2\i) -- (w3\i) --(w4\i) -- (x5\i);
}
\draw[very thick,color=red] (x41) -- (x44);
\draw[very thick,color=red] (x31) -- (x34);
\draw[very thick,color=red] (x21) -- (x24);
 \node(c) at (-5.5+14,-7.5) {$\SoP$};
 
\end{tikzpicture}
   }
    \end{center}
    \caption{Square $4$-by-$4$ parallel-of-series, hammocks and series-of-parallel.}\label{fig:Hammock}
    \end{figure}

\paragraph{Matchstick Minimal Networks}
Let $l$ and $w$ be two strictly positive integers. A $\TN$ $\Nss$ is a MMN if and only if it can be designed in one of the following two ways. Either start by a parallel-of-series ($\PoS$) of width $w$ and length $l$ and place vertical matchsticks arbitrarily; or start with a series-of-parallel ($\SoP$) of width $w$ and length $l$ and remove vertical matchsticks arbitrarily.

Another way of defining a MMN, described in \cite{DB19}, is by using the bijection between the set of all MMNs of length $l$ and width $w$ and the set of all binary matrices $M_{\Nss}\in \mathcal{M}_{(l-1)\times(w-1)}{\{0,1\}}$. At any $\lwMMN$ $\Nss$ we associate its matchstick incidence matrix $M_{\Nss}\in \mathcal{M}_{(l-1)\times(w-1)}{\{0,1\}}$, as
\begin{itemize}
\item $M_{\Nss}(i,j)=1$ if there is a matchstick at position $(i,j)$;
\item $M_{\Nss}(i,j)=0$ if there is no matchstick at position $(i,j)$.
\end{itemize}   
\paragraph{Hammock networks} MMNs presenting a ``brick-wall'' pattern are known as hammocks \cite{1956_MS_1, 1956_MS_2, 2018_CBDP,BD2020}. Starting from a $\SoP$, by alternately deleting matchsticks, one can construct a hammock. If $w$ and $l$ are both even there are two hammock ($\Hs$ and $\Hs^{+}$), while otherwise only one hammock exists $\Hs.$
Using the matchstick incidence matrix we have $M_{\PoS}=\bm{0}_{(w-1)\times(l-1)}$ and $M_{\SoP}=\bm{1}_{(w-1)\times(l-1)}$ (see Figure \ref{fig:Hammock}).

\paragraph{\textbf{Duality properties}}
Let $\Nss$ be a MMN. The dual of $\Nss$, denoted here by $\Nss^{\bot}$, was introduced in \cite{1956_MS_1}.  
Duality properties were proved in \cite{DJ2019,CD2020} in case of particular MMNs, such as hammocks. Some duality properties, that are needed in the context of approximations (as mentioned in \cite{CD2020}), are recalled in the next subsection. Let us denote by $\bm{1}_{l\times w}$ the all-ones matrix, and the bit-wise complement of a binary matrix $M_{\Nss}\in \mathcal{M}_{(l-1)\times(w-1)}{\{0,1\}}$ as
\begin{equation}
\overline M_{\Nss}=\bm{1}_{(l-1)\times (w-1)}\oplus M_{\Nss},
\end{equation} 
 where $\oplus$ denotes the XOR operation .
\begin{theorem}[\cite{DB19}]\label{thm:dual-mmn}
Let $\Nss$ be a $\lwMMN$. Then either $l=1$ ($\Nss$ being the all parallel network) and we have $\Nss^{\bot}$ is the all series network, or $w,l\ge 2$ and we have $M_{\Nss^{\bot}}=\left(\overline M_{\Nss}\right)^{t}.$
\end{theorem}

Notice that by Theorem \ref{thm:dual-mmn}, the dual of a $\lwMMN$ is a $\wlMMN$.

\subsection{Reliability polynomial}\label{subsec:Rel-poly}

The reliability of a $\TN$ is defined as the probability that the source $S$ and the terminus $T$ are connected, given that each device closes with probability $p.$ The reliability polynomial is presented in the literature under several forms, depending on the basis of the linear space of polynomials that is taken into account. If the Bernstein basis $$\left\{\binom{n}{k}p^{k}(1-p)^{n-k}| k\in \{0,1,...,n\}\right\}$$ 
is used, then we have the so-called N-form (see \cite{1956_MS_1})
\begin{equation} 
\Rel(\Nss;p)=\sum\limits_{k=0}^nN_k\;p^k(1-p)^{n-k}.\label{rel:N_form}
\end{equation}
The coefficient $N_k$ represents the number of ways one can select a subset of $k$ devices in $\Nss$ such that if these $k$ devices are closed and the remaining are open, then the two terminals $S$ and $T$ are connected, i.e., $\Nss$ is closed. Straightforward, well-known basic properties of $N_k$ can be immediately deduced from the definition.

\begin{property}[\cite{1956_MS_1}]\label{pr:Nk-prop}
If $\Nss$ is a $\lwTN$ then:
\begin{itemize}
    \item $\forall k \in \{0,\dots,n\}\;,\; 0\le N_k\leq \binom{n}{k}$; 
    \item $\forall k \in\{0,\dots,l-1\}\;,\;N_k=0$;
    \item $\forall k \in\{n-w+1,\dots,n\}\;,\;N_k=\binom{n}{k}.$
\end{itemize}
\end{property}

\begin{corollary}[\cite{CD2020}]
Denoting by $a_{k}$ the coefficient of $\Rel(\Nss;p)$ written in Bernstein basis, then the coefficients in \eqref{rel:N_form} are $N_k=\binom{n}{k}a_{k}$. As consequence, we deduce that $0\leq a_{k}\leq 1$ for all $k\in \{0,1,...,n\}.$ 
\end{corollary}

In the sequel, we consider two dual MMNs, denoted by $\Nss$ and $\Nss^{\bot}.$ All over the paper, the coefficients of the reliability polynomial of $\Nss$ will be denoted as in \eqref{rel:N_form} and the coefficients of the reliability polynomial of the dual network $\Nss^{\bot}$ will be denoted by $N_k^{\bot}$. The following complementarity property is well known:
 
 \begin{property}[\cite{1956_MS_1}] 
 If $\Rel(\Nss;p)$ and $\Rel(\Nss^{\bot};p)$ are the reliability polynomials of two dual MMNs of type $(l,w)$ respectively $(w,l)$, then
\begin{equation}\label{eq:dual_polyn}
\Rel(\Nss;p)+\Rel(\Nss^{\bot};1-p)=1.
\end{equation}
 \end{property}

Equation \eqref{eq:dual_polyn} leads to the following complementarity identity, proved in \cite{CD2020} (Property 2 pp.80) for hammock networks:

 \begin{property}[\cite{CD2020,CD2021}] If $N_{k}$ and $N_{k}^{\bot}$ are coefficients of the reliability polynomials of two dual MMNs, $\Nss$ and $\Nss^{\bot}$, then
\begin{equation}\label{suCo}
N_{k}+N_{n-k}^{\bot} = \binom{n}{k},
\end{equation}
for all $k\in \{0,1,2,...,n\}.$
\end{property}
Identity \eqref{suCo} also holds for any dual 2TNs (the proof is identical).

\paragraph{Parallel-of-series and Series-of-parallel.}
For some type of MMNs there is a closed formula of the reliability polynomial. 
\begin{theorem}[\cite{CD2021a}]\label{thm:rel_POS} Let $\PoS$ be a $\lwMMN.$ Then we have

\begin{equation}
\Rel\left(\PoS;p\right)=\sum\limits_{k=l}^{n}\sum\limits_{j=1}^{\left[\frac{k}{l}\right]}(-1)^{j+1}\binom{w}{j}\binom{n-jl}{n-k}p^k(1-p)^{n-k},
\end{equation}
where $\left[\frac{k}{l}\right]$ denotes the integer part of the fraction.
\end{theorem}
The proof relies of the fact that one can write the coefficients of a $(l,w)-\PoS$ using the formula
\begin{equation}\label{coef:PoS}
N_k=\sum\limits_{j=1}^{\left[\frac{k}{l}\right]}(-1)^{j+1}\binom{w}{j}\binom{n-jl}{n-k}.
\end{equation}

Combined with \eqref{suCo} one can deduce 
\begin{equation}
    N_{n-k}^{\bot}=\sum\limits_{j=0}^{\left[\frac{k}{l}\right]}(-1)^{j}\binom{w}{j}\binom{n-jl}{n-k}.
\end{equation}

\begin{remark}
The coefficients of a $\PoS$ also have a combinatorial interpretation, fact that allows one to deduce basic properties such as those in Proposition \ref{pr:Nk-prop}. Indeed, $N_k$ represents the number of ways one can distribute $k$ balls among $w$ urns, where the urns have height $l$, such that at least one urn is completely filled with balls. 
\end{remark}

Using the complementary property induced by duality \eqref{suCo} we deduce the following result.

\begin{proposition}
Let $\Nss$ be a $\lwMMN.$ Then we have 
\begin{equation}
\sum\limits_{j=1}^{\left[\frac{k}{l}\right]}(-1)^{j+1}\binom{w}{j}\binom{n-jl}{n-k} \leq N_k\leq  \sum\limits_{j=0}^{\left[\frac{n-k}{w}\right]}(-1)^{j}\binom{l}{j}\binom{n-k-jw}{k}
\end{equation}
\end{proposition}

\begin{proof}
Any coefficient $N_k$ of an arbitrary $\lwMMN$ is bigger that or equal to the same coefficient of an $(l,w)-\PoS$, and smaller than or equal to the same coefficient of an $(l,w)-\SoP.$ Using the formula \eqref{coef:PoS} one obtains the first inequality. As for the second inequality we the formula for the coefficients of a $(w,l)-\PoS$, i.e., $\sum_{j=1}^{\left[\frac{k}{w}\right]}(-1)^{j+1}\binom{l}{j}\binom{n-jw}{n-k}$, combined with the fact that its dual is a $(l,w)-\SoP.$ This yields that the $n-k$ coefficient of an $(l,w)-\SoP$ equals $\sum_{j=0}^{\left[\frac{k}{w}\right]}(-1)^{j+1}\binom{l}{j}\binom{n-jw}{n-k}.$ Hence, a simple variable change implies the wanted result.

\end{proof}
\section{Mutual shape properties of the reliability polynomials of two dual networks}\label{sec:Shape-Prop}

\subsection{Convexity of high order}\label{subsec:convex-HO}

Let us consider $[a,b]\subseteq \mathbb{R}$ an interval and a function $f:[a,b]\rightarrow \mathbb{R}$. Suppose that $n\in \mathbb{N}$.
\begin{definition}
The divided difference of order $n$ of function $f$ on points $a\leq x_1 < x_2 < ... < x_{n+1} \leq b$ is the number defined by:
\begin{equation}
    [x_1, x_2, ..., x_{n+1};f]=\frac{ [x_2, x_3, ..., x_{n+1};f]- [x_1, x_2, ..., x_{n};f]}{x_{n+1}-x_1},
\end{equation}
\begin{equation*}
    [x_1;f]=f(x_1).
\end{equation*}
\end{definition} 

\begin{remark}\label{dnDD}
It is known (see, for example \cite{P1945}) that if a function $f$ is $n$-th order differentiable on a point $x\in (a,b)$ then then the derivative $f^{(n)}(x)$ equals to the limit of $n![x_1, x_2, ..., x_{n+1};f]$ when all points $x_i, i\in \{1,2,...,n+1\}$ tend to $x$.
\end{remark}

The concept of convex function of high order on an interval was introduced in 1926 by E. Hopf \cite{Hopf}. T. Popoviciu \cite{P1934} extended this concept to functions defined on an arbitrary set in 1934. Also, T. Popoviciu \cite{P1945} extensively studied this concept in case of real functions of several real variables.

\begin{definition} [\cite{P1945}]
Function $f$ is said to be $n$-th order convex (non-concave, polynomial, non-convex, concave) on $[a,b]$ if 
\begin{equation}\label{nCV}
[x_1, x_2, ..., n+2;f] > (\geq , =, \leq, <)\quad 0, 
\end{equation}
respectively, for all systems of points $a\leq x_1 < x_2 < ... < x_{n+2} \leq b$
\end{definition}

The functions having one of the properties defined by means of \eqref{nCV} are generally called $n$-th order functions (see \cite{P1934}, \cite{P1935}).

\begin{remark}
The reliability polynomial of a network of type $n$ is a $n$-th order polynomial function on $[0,1]$, since every $n$-th degree polynomial has this property \cite{Hopf}.
\end{remark}

\begin{remark}\label{DkCv}
If function $f$ is $n+1$-th order differentiable on $[a,b]$ then, in view of Remark \ref{dnDD}, it follows that condition \eqref{nCV} is expressible in terms of derivatives as follows: function $f$ is $n$-th order convex (non-concave, polynomial, non-convex, concave) on $[a,b]$ if 
\begin{equation}
f^{(n+1)}(x) > (\geq , =, \leq, <)\quad 0,   
\end{equation}
respectively, for all $x\in [a,b].$ The one side derivatives are considered on points $a$ and $b.$
\end{remark}

\subsection{Convexity properties of the reliability polynomials of two dual minimal networks}\label{subsec:convex-prop-MMNs}
Properties of convexity are accidentally mentioned, both in case of the reliability polynomial of a MMN (see \cite{CD2020}, \cite{CD2021}) and in case of its coefficients sequence (see \cite{H2015}). This subsection presents the research results on the presence of various types of high order convexity to the reliability polynomials of two dual MMNs. The impact of networks duality on the shape of the polynomials is emphasized. 

\begin{figure}
\begin{subfigure}{.32\textwidth}
  \centering
  \includegraphics[width=\linewidth]{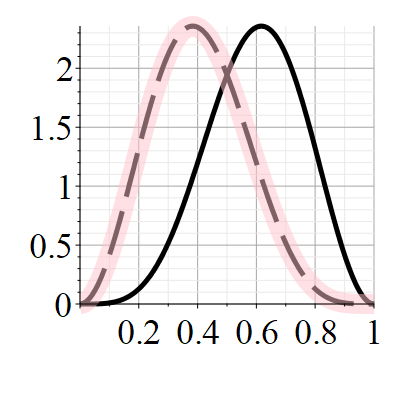}
  \caption{$k=1$}
  \label{fig:sub-first}
\end{subfigure}
\begin{subfigure}{.32\textwidth}
  \centering
  \includegraphics[width=\linewidth]{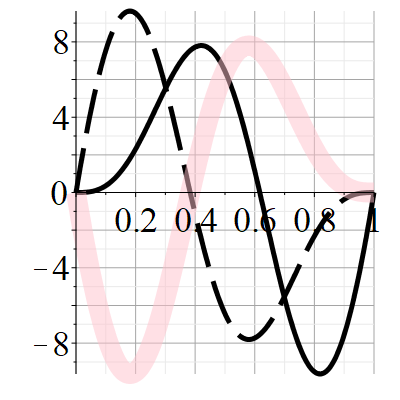}
  \caption{$k=2$}
  \label{fig:sub-second}
\end{subfigure}
\begin{subfigure}{.32\textwidth}
  \centering
  \includegraphics[width=\linewidth]{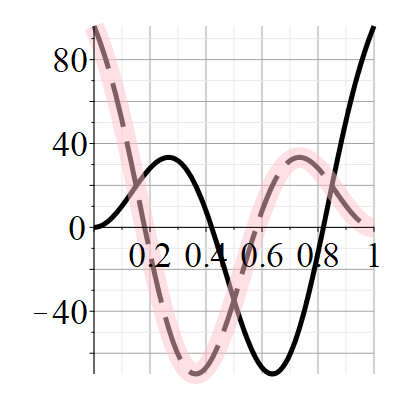}
  \caption{$k=3$}
  \label{fig:sub-third}
\end{subfigure}
\caption{$\frac{d^k}{dp^k}\Rel(\Nss;p)$ (dash black line), $\frac{d^k}{dp^k}\Rel(\Nss^{\bot};p)$ (solid black line), and $\frac{d^k}{dp^k}\Rel(\Nss^{\bot};1-p)$ (solid pink line) for the 3-by-5 hammock network}
\label{fig:deriv}
\end{figure}

\begin{theorem}\label{kCv-Rel}
If $\Rel(\Nss;p)$ and $\Rel(\Nss^{\bot};p)$ are the reliability polynomials of two dual MMNs of type $(l,w)$ respectively $(w,l)$, $n=lw \geq 3,$ and $0 \leq k \leq n$ then the following holds.
\begin{enumerate}
    \item If $k$ is odd then $\Rel(\Nss;p)$ and $\Rel(\Nss^{\bot};1-p)$ are $(k-1)$-th order functions of the same type on each sub-interval of $[0,1]$: either both are $(k-1)$-th order convex or both are $(k-1)$-th order concave.
    \item If $k$ is even then $\Rel(\Nss;p)$ and $\Rel(\Nss^{\bot};1-p)$ are $(k-1)$-th order functions of opposite types on each sub-interval of $[0,1]$: if one polynomial is $(k-1)$-th order convex then the other one is $(k-1)$-th order concave, and conversely.
\end{enumerate}
\end{theorem}
\begin{proof}
The two polynomials are differentiable functions of all orders, verifying \eqref{eq:dual_polyn}. By successively differentiating the equation \eqref{eq:dual_polyn} $k$ times one gets:
\begin{equation}\label{DkRel}
    \frac{d^k}{dp^k}\Rel(\Nss;p) = (-1)^{k+1} \frac{d^k}{dp^k}\Rel(\Nss^{\bot};1-p),
\end{equation}
which gives the conclusion, based on Remark \ref{DkCv}.
\end{proof}

Figure \ref{fig:deriv} shows the mutual behaviour of the derivatives of the reliability polynomials of two dual hammock networks. The corresponding impact on the shape of reliability  polynomials is described by the following corollaries.

\begin{corollary}
Suppose that $n\geq 3$ and $k=2$. Relation \eqref{DkRel} implies that if $p_{0}$ is an inflection point of the polynomial $\Rel(\Nss;p)$ then it is an inflection point of $\Rel(\Nss^{\bot};1-p),$ which implies that $1-p_{0}$ is an inflection point of $\Rel(\Nss^{\bot};p).$
\end{corollary}

\begin{corollary}
Suppose that $n\geq 3,$ $k$ is odd and take $p=\frac{1}{2}$ in equation \eqref{DkRel}. It follows that 
$$\frac{d^k}{dp^k}\Rel(\Nss;\frac{1}{2}) = \frac{d^k}{dp^k}\Rel(\Nss^{\bot};\frac{1}{2}),$$
which means that all derivatives of odd order of the reliability polynomials of two dual networks, $\Rel(\Nss;p)$ and $\Rel(\Nss^{\bot};p)$, have the same value. In particular, if $k=1$ one gets that the two polynomials, $\Rel(\Nss;p)$ and $\Rel(\Nss^{\bot};1-p)$, have parallel tangents at $p=\frac{1}{2}.$
\end{corollary}

\begin{corollary}
Suppose that $n\geq 3,$ $k$ is even and take $p=\frac{1}{2}$ in equation \eqref{DkRel}. It follows that 
$$\frac{d^k}{dp^k}\Rel(\Nss;\frac{1}{2}) = - \frac{d^k}{dp^k}\Rel(\Nss^{\bot};\frac{1}{2}).$$
In particular, if $k=2$ one gets that if $p=\frac{1}{2}$ is an inflection  point of a reliability polynomial of a network then it is an inflection point of the reliability polynomial of the dual network as well.
\end{corollary}

\begin{corollary}
Suppose that $\Nss$ is a minimal network of length $l\geq 3$ and width $w$, which means that $n=lw\geq 3$. Property \ref{pr:Nk-prop} implies that all derivatives of order $k < l$ of its reliability polynomial $\Rel(\Nss;p)$ equal to $0$ at $p=0$. As consequence, the x-axis is tangent to the graph of this polynomial at $p=0$ and the curvature radius of $\Rel(\Nss;p)$ tends to infinity when $p\rightarrow 0$. All these, together with the non-negativity of a reliability polynomial over $[0,1]$, imply that the polynomial is first order convex in the neighborhood of the origin. Theorem \ref{kCv-Rel} implies that the reliability polynomial $\Rel(\Nss^{\bot};1-p)$ is concave in the neighborhood of the origin. The functions $\Rel(\Nss^{\bot};1-p)$ and $\Rel(\Nss^{\bot};p)$ are symmetric with respect to the straight line $p=\frac{1}{2}$, which implies that it is also convex in the neighborhood of the origin.
\end{corollary}

Another convexity property of functions that is important in the context of the reliability theory in case of $\TN$s is the log-convexity, defined as follows:
\begin{definition}[\cite{W1994}]
A function $f:E\rightarrow \mathbb{R}$, $E\subseteq \mathbb{R}$, is said to be log-convex (or log-concave) on $E$ if $f(x)>0$ for all $x\in E$ and function $\log(f)$ is convex (or concave, respectively) on $E$.
\end{definition}
\begin{remark}\label{logconc}
It is proved in \cite{W1994}, pp. 207, that the log-convex (log-concave) functions are also convex (concave) functions, but the converse is not true. Based on the results of Huh \cite{H2015} and Lenz \cite{L2013} the sequence of the coefficients of a the reliability polynomial is log-concave. This property will be used in the sequel, because it implies the concavity of the coefficients function of a reliability polynomial, which will be defined in the next section.
\end{remark}

Consider $\Nss$ a $\lwMMN$. The reliability polynomial of this network, expressed in Bernstein basis, is $\Rel(\Nss;p)$ defined by (\ref{rel:N_form}). Knowing the reliability polynomial $\Rel(\Nss;p)$ is equivalent to knowing the corresponding function $F_{(l,w)}$ defined by (\ref{F}). We consider the dual network $\Nss^{\bot}$ together with its reliability polynomial $\Rel(\Nss^{\bot};p)$. Let us define two functions, which we call coefficients functions in the sequel:
$F_{(l,w)}:[0,n]\rightarrow \mathbb{R}$ and $F_{(w,l)}:[0,n]\rightarrow \mathbb{R}$ by
\begin{equation}\label{F}
F_{(l,w)}(x)=
\begin{cases}
0, & \text{if $x=0$}\\
(N_k-N_{k-1})x+kN_{k-1}-(k-1)N_k, & \text{if $x\in [k-1,k], k\in\{1,2,...,n\}$}.
\end{cases}
\end{equation}
\begin{equation}\label{Fd}
F_{(w,l)}(x)=
\begin{cases}
0, & \text{if $x=0$}\\
(N_k^{\perp}-N_{k-1}^{\perp})x+kN_{k-1}^{\perp}-(k-1)N_k^{\perp}, & \text{if $x\in [k-1,k], k\in\{1,2,...,n\}$}.
\end{cases}
\end{equation}

In fact, function $F_{(l,w)}$ (respectively $F_{(w,l)}$) is the segmentary linear function obtained based on the coefficient functions of the reliability polynomials two  dual hammock networks, as defined in \cite{H2015,L2013,CD2021}. Knowing the reliability polynomials of the two dual networks is equivalent to knowing the two coefficient functions defined by (\ref{F}) and (\ref{Fd}).
\begin{remark}
The coefficient function $F_{(l,w)}$ is concave on $[l-1, n-w+1]$. The coefficients function $F_{(w,l)}$ is concave functions on $[w-1, n-l+1]$. This shape is a consequence of Property \ref{pr:Nk-prop} and Remark \ref{logconc}.
\end{remark}

\subsection{Extremal properties of the coefficients functions}
The sequence of coefficients of the reliability polynomial have some monotony properties, that are consequences both of their complementarity property \eqref{suCo} and of Remark \ref{logconc}. The concavity of functions $F_{(l,w)}$ and $F_{(w,l)}$ on $[0,n]$, together with Property \ref{pr:Nk-prop}, imply that the index of the maximum coefficient of the two reliability polynomials is in $[l-1, n-w+1]$, and $[w-1, n-l+1]$ respectively. We construct, in the next section, a method to approximate the reliability polynomials of two dual networks, denoted here by $\Nss$ of $(l,w)$-type, and its dual $\Nss^{\bot}$ of $(w,l)$-type. We prove that the maximum coefficient of the approximate reliability polynomial of the $\Nss$ is reached in the same interval as the maximum coefficient of the exact polynomial. The most frequent interval that contains the index of the maximum point of $F_{(l,w)}$, as identified by studying the completely known reliability polynomials of small size hammock networks and also small size compositions of series and parallel, is 
$$I_1=\left[\frac{n-w+l}{2}, \frac{n-w+l}{2}+\frac{n-w-l+2}{4}\right].$$
Few networks have the index of the maximum coefficient of the reliability polynomial not belonging to $I_1$ but to a larger interval,
$$I_2=\left[\frac{n-w+l}{2}, n-w+\frac{1}{2}\right].$$ 
Few examples are presented in Table \ref{tab:my_label}. The maximum coefficient of the reliability polynomial is denoted by $\max(N_{k})$, and the value of the index $k$ of the maximum coefficient is denoted by $\argmax(F_{(l,w)}(x))$ in this table.

\begin{table}[!ht]
    \centering
    \begin{tabular}{|c|c||c|c|c|c|}
    \hline
         $w$&$l$&$\max(N_{k})$&$\argmax(F_{(l,w)}(x))$&$I_1$&$I_2$  \\
         \hline
     \multirow{4}{*}{2}&3&10&4&$[3.5,4]$&$[3.5,4.5]$\\
     \cline{2-6}
         &4&20&6&$[5,6]$&$[5,6.5]$\\
        &4&24&5&$[5,6]$&$[5,6.5]$\\
        \cline{2-6}
        &5&56&7&$[6.5,8]$&$[6.5,8.5]$\\
        \hline
        \multirow{4}{*}{3}&2&16&3&$[2.5,3.25]$&$[2.5,3.5]$\\
        \cline{2-6}
        &3&84&5&$[4.5,5.75]$&$[4.5,6.5]$\\
        \cline{2-6}
        &4&450&7&$[6.5,8.25]$&$[6.5,9.5]$\\
        \cline{2-6}
        &5&2443&9&$[8.5,10.75]$&$[8.5,12.5]$\\
        \hline
        \multirow{6}{*}{4}&2&62&4&$[3,4]$&$[3,4.5]$\\
        &2&66&4&$[3,4]$&$[3,4.5]$\\
        \cline{2-6}
        &3&698&7&$[5.5,7.25]$&$[5.5,8.5]$\\
        \cline{2-6}
        &4&7700&9&$[8,10.5]$&$[8,12.5]$\\
        &4&8312&9&$[8,10.5]$&$[8,12.5]$\\
        \cline{2-6}
        &5&88948&11&$[10.5,13.75]$&$[10.5,16.5]$\\
        \hline
        \multirow{4}{*}{5}&2&244&5&$[3.5,4.75]$&$[3.5,5.5]$\\
        \cline{2-6}
        &3&5653&8&$[6.5,8.75]$&$[6.5,10.5]$\\
        \cline{2-6}
        &4&132750&11&$[9.5,12.75]$&$[9.5,15.5]$\\
        \cline{2-6}
        &5&3162650&14&$[12.5,16.75]$&$[12.5,20.5]$\\
        \hline
    \end{tabular}
        \caption{Extremal values of the coefficients functions for small hammocks.}
    \label{tab:my_label}
\end{table}

A special case is presented by Parallel-of-series and Series-of-parallel networks. In case of a $\PoS$, we proved that index of the maximum value of the coefficients \eqref{coef:PoS} is in $\left[\frac{n}{2}, n-w+\frac{1}{2}\right].$ We omit the proof of this property both because of its length and because it exceeds the purpose of this paper. But we remark that the coefficients of the reliability polynomial of all types of network have the same extremal property, which we retrieve to its approximant.

\section{Shape preserving simultaneous approximation of the reliability polynomials of two dual two-terminal networks}\label{sec:Approx-algo}

\subsection{An efficient constructive method}\label{subsec:alg-approx}
In his section we intend to build a method of approximation of functions $F_{(l,w)}$ and $F_{(w,l)}$ by means of a spline function, starting from the properties of the reliability polynomials described above and in \cite{2017_CDBP} and \cite{2018_CBDP}. Some generalized convexity properties as described in \cite{CL2002} will be used. We construct segmentary polynomial function meant to imitate the shape of functions $F_{(l,w)}$ and $F_{(w,l)}$. As proved in \cite{P1935}, given a continuous function on a bounded closed interval, the Bernstein approximation polynomial of degree $s$ of this function preserves the convexity of the approximated function (see also \cite{L1953} and \cite{Nat1964}). This property gave us the idea of approximating functions $F_{(l,w)}$ and $F_{(w,l)}$ by means of polynomials imitating the Bernstein polynomial of third degree in \cite{CD2020}. The results of the cubic spline approximation algorithm presented in \cite{CD2020} constructed both by means of Lagrange interpolation and by a weakened Bernstein type approximation operator are compared in \cite{CD2021}. Also, the algorithm from \cite{CD2020} is refined in \cite{CD2021} in order to improve the accuracy of the approximation. In this paper we describe a version of the approximation algorithm from \cite{CD2021} obtained by replacing the cubic splines with quadratic splines. The initial information on the two coefficients functions $F_{(l,w)}$ and $F_{(w,l)}$ refers to their values on intervals $[0, l-1] \cup [n-w+1, n]$ and $[0, w-1]\cup [n-l+1, n]$ respectively. We also have shape information on these functions, i.e. the concavity of these functions is a consequence of the results from \cite{H2015} and \cite{L2013}.
We complete the missing information from intervals $[l-1, n-w+1]$ and respectively $[w-1, n-l+1]$ by the known information on the shape of the coefficients functions $F_{(l,w)}$ and $F_{(w,l)}$. This is the reason to carefully chose the initial knots in order to generate a function having the same shape as $F_{(l,w)}$ and $F_{(w,l)}$. In order to approximate $F_{(l,w)}$ and $F_{(w,l)}$ we construct two continuous quadratic spline functions $f_{(l,w)}:[0,n]\rightarrow \mathbb{R}$ and $f_{(w,l)}:[0,n]\rightarrow \mathbb{R}$ that verify the following conditions:

\begin{equation}
\left\{
  \begin{array}{ll}
    f_{(l,w)}(0)=f_{(l,w)}(1)=...=f_{(l,w)}(l-1)=0 \\
    f_{(l,w)}(s)=N_{s}>\binom{n}{w-1} \\
    f_{(l,w)}(n-w+k)=\binom{n}{w-k}, k\in \{0,1,...w-1\},
  \end{array}
\right.
\end{equation}

\begin{equation}
\left\{
  \begin{array}{ll}
    f_{(w,l)}(0)=f_{(w,l)}(1)=...=f_{(w,l)}(w-1)=0 \\
    f_{(w,l)}(t)=N_{t}^{\perp}>\binom{n}{l-1} \\
    f_{(w,l)}(n-l+k)=\binom{n}{l-k}, k\in \{0,1,...l-1\}.
  \end{array}
\right.
\end{equation}
for some points $s\in [l-1, n-w+1]$ and $t\in [w-1, n-l+1]$. 
\begin{remark}
If $l>2$ and $w>2$, it is always possible to find two numbers $s\in [l-1, n-w+1]$ and $t\in [w-1, n-l+1]$ such as the two conditions $N_{s}>\binom{n}{w-1}$ and $N_{t}^{\perp}>\binom{n}{l-1}$ are valid. Indeed, one can always compute $N_{l}$ and $N_{w}^{\perp}$ by means of the technique from \cite{2018_CBDP}. Then, one can compute $N_{n-w}$ and $N_{n-l}^{\perp}$ using the coefficients complementarity relation \eqref{suCo}. The relation\eqref{suCo} implies that at least two of the four coefficients verify the needed conditions. 
\end{remark}

In order to define the two functions $f_{(l,w)}$ and $f_{(w,l)}$ we  have previously taken into account the convexity properties of the second degree polynomial, that allows us to define approximation operators that preserve some shape properties of the approximated curve. In the sequel we define an approximation function by interpolating the coefficients functions using quadratic splines conveniently chosen in order to preserve the convexity and concavity shapes. Function $f_{(l,w)}$ is searched as:
\begin{equation}\label{flw}
f_{(l,w)}(x)=
\begin{cases}
0, & \text{if $0\leq x\leq l-1$}\\
Ax^{2}+Bx+C, & \text{if $l-1<x\leq n-w+1$}\\
d_{(l,w)}(k)(x), & \text{if $x\in (k-1, k], k\in\{n-w+2,...,n\}$}
\end{cases}
\end{equation}
Here
\begin{align*}
d_{(l,w)}(k)(x)&=\left(\binom{n}{k}-\binom{n}{k-1}\right)x+k\binom{n}{k-1} - (k-1)\binom{n}{k}
\end{align*}
are the straight line segments determined by points $(k-1, \binom{n}{k-1})$ and $(k, \binom{n}{k})$, for all $k\in\{n-w+2,...,n\}$ respectively. Also, the coefficients $A, B, C\in \mathbb{R}$ are obtained using the interpolation conditions:

\begin{equation}\label{inter}
\begin{cases}
\lim_{\substack{x\rightarrow l-1\\x>l-1}} f_{(l,w)}(x)=0\\
f_{(l,w)}(s)=N_{s}>\binom{n}{w-1}\\
f_{(l,w)}(n-w+1)=\binom{n}{w-1}
\end{cases}
\end{equation}

Function $f_{(w,l)}$ is searched as:
\begin{equation}\label{fwl}
f_{(w,l)}(x)=
\begin{cases}
0, & \text{if $0\leq x\leq w-1$}\\
A^{\perp}x^{2}+B^{\perp}x+C^{\perp}, & \text{if $w-1<x\leq n-l+1$}\\
d_{(w,l)}(k)(x), & \text{if $x\in (k-1, k], k\in\{n-l+2,...,n\}$}
\end{cases}
\end{equation}

Here
\begin{align*}
&d_{(w,l)}(k)=\left(\binom{n}{k}-\binom{n}{k-1}\right)x+k\binom{n}{k-1} - (k-1)\binom{n}{k}
\end{align*}
are the straight line segments determined by points $(k-1, \binom{n}{k-1})$ and $(k, \binom{n}{k})$, for all $k\in\{n-l+2,...,n\}$. As above, the coefficients $A^{\perp}, B^{\perp}, C^{\perp} \in \mathbb{R}$ are obtained using the interpolation conditions:

\begin{equation}\label{interD}
\begin{cases}
\lim_{\substack{x\rightarrow w-1\\x>w-1}} f_{(w,l)}(x)=0\\
f_{(w,l)}(t)=N_{t}^{\perp}>\binom{n}{l-1}\\
f_{(w,l)}(n-l+1)=\binom{n}{l-1}
\end{cases}
\end{equation}

The interpolation conditions \eqref{inter} and \eqref{interD} lead to the following systems of linear equations in order to compute the functions $f_{(l,w)}$ and $f_{(w,l)}$ using (\ref{flw}) and (\ref{fwl}):

\begin{equation}
\left\{
  \begin{array}{ll}
    A(l-1)^{2}+B(l-1)+C=0 \\
    As^{2}+Bs+C=N_{s}\\
    A(n-w+1)^{2}+B(n-w+1)+C=\binom{n}{w-1},
  \end{array}
\right.\label{eq:syst21}
\end{equation}

\begin{equation}
\left\{
  \begin{array}{ll}
    A^{\perp}(w-1)^{2}+B^{\perp}(w-1)+C^{\perp}=0 \\
    A^{\perp}t^{2}+B^{\perp}t+C^{\perp}=N_{t}^{\perp}\\
    A^{\perp}(n-l+1)^{2}+B^{\perp}(n-l+1)+C^{\perp}=\binom{n}{l-1}.
  \end{array}
\right.\label{eq:syst22}
\end{equation}

 The approximation algorithm, based on determining the functions $f_{(l,w)}$ and $f_{(w,l)}$ using the solutions of the two systems of equations obtained by Cramer's rule, is as follows.\\
\clearpage
\textbf{The algorithm}:
\begin{description}
  \item[Step 1] Compute the values of two coefficients $N_{s}$ and $N_{t}^{\perp}$ using some technique from literature. If $N_{l}$ and $N_{w}$ are chosen then we use the method from \cite{2018_CBDP} and then we compute the values $N_{n-w}$ and $N_{n-l}^{\perp}$ using (\ref{suCo}). After that we put $s=n-w$ and $t=n-l$.
  \item[Step 2] Compute the coefficients of the approximate functions $f_{(l,w)}$ and $f_{(w,l)}$, by:
  \begin{align*}
      A&=\frac{\binom{n}{w-1}(s-l+1)-N_{s}(n-w-l+2)}{(s-l+1)(n-w-l+1)(n-w-l+2)}\\
      B&=\frac{N_{s}(n-w-l+2)(n-w+l)-\binom{n}{w-1}(s-l+1)(s+l-1)}{(s-l+1)(n-w-l+1)(n-w-l+2)}\\
      C&=\frac{(l-1)\left[\binom{n}{w-1}s(s-l+1)-N_{s}(n-w+1)(n-w-l+2)\right]}{(s-l+1)(n-w-l+1)(n-w-l+2)}\\
      A^{\perp}&=\frac{\binom{n}{l-1}(t-w+1)-N_{t}^{\perp}(n-w-l+2)}{(t-w+1)(n-w-l+1)(n-w-l+2)}\\
      B^{\perp}&=\frac{N_{t}^{\perp}(n-w-l+2)(n-l+w)-\binom{n}{l-1}(t-w+1)(t+w-1)}{(t-w+1)(n-w-l+1)(n-w-l+2)}\\
      C^{\perp}&=\frac{(w-1)\left[\binom{n}{l-1}t(t-w+1)-N_{t}^{\perp}(n-l+1)(n-w-l+2)\right]}{(t-w+1)(n-w-l+1)(n-w-l+2)}.
  \end{align*}

  \item[Step 3.] Write functions $f_{(l,w)}$ and $f_{(w,l)}$ using (\ref{flw}) and (\ref{fwl}) respectively.
  \item[Step 4.] Compute $f_{(l,w)}(k)$, $k\in \{l, l+1,..., n-w\}.$ 
  \item[Step 5.] Compute $f_{(w,l)}(k)$, $k\in \{w, w+1,..., n-l\}.$ 
  \item[Step 6.] Compute 
  \begin{equation}\label{delta}
  \Delta(k)=\binom{n}{k}-f_{(l,w)}(k)-f_{(w,l)}(n-k),
  \end{equation} 
  for each $k\in \{\min\{l-1,w-1\},...\max\{n-l+1,n-w+1\}\}$.
  \item[Step 7.] Compute $\tilde{N}(F_{(l,w)};k)=f_{(l,w)}(k)+\frac{\Delta(k)}{2}$, and  $\tilde{N}(F_{(w,l)};n-k)=f_{(w,l)}(n-k)+\frac{\Delta(k)}{2}$. 
  \item[Step 8.] If $\tilde{N}(F_{(l,w)};k)<0$ then replace $\tilde{N}(F_{(l,w)};k)=0$, and put  $\tilde{N}(F_{(w,l)};n-k)=\binom{n}{k}$ (or converse, if the dual coefficient is negative). 
  
  \item[Step 9.] Output the approximation polynomials 
\begin{equation}
\ApRel(\Nss;p)=\sum\limits_{k=0}^{n}\tilde{N}(F_{(l,w)};k)\;p^k(1-p)^{n-k},\label{rel:approx-pol}\end{equation}
\begin{equation}
\ApRel(\Nss^{\bot};p)=\sum\limits_{k=0}^{n}\tilde{N}(F_{(w,l)};k)\;p^k(1-p)^{n-k}.\label{rel:approx-poldual}
\end{equation}
\end{description}

\begin{remark}
The overall time complexity of computing $\ApRel(\Nss;p)$ and $\ApRel(\Nss^{\bot};p)$ is $O(n)$, when $n\to\infty.$ This fact can be easily deduce either by inspecting each step of our algorithm or by adapting the proof of Theorem 5 from \cite{CD2021}. 
\end{remark}

\subsection{Shape and extremum properties of the approximation operator}\label{subsec:properties-approx}

In this subsection we suppose that both $l\geq 2$ and $w\geq 2$ and at least one inequality is strict. It implies that $n\geq 6$. 
\begin{remark}
All the invariant properties proved in \cite{CD2021} (Property 14, Corollary 15 and Property 16) in case of the use of an approximation operator constructed by means of cubic spline functions stay valid. One can prove that the approximate reliability polynomial of two dual networks obtained by using quadratic splines keep invariant the complementarity relations \eqref{eq:dual_polyn}, \eqref{suCo} and their consequences in a similar manner as in \cite{CD2021}, which means that:
\begin{equation}
    \tilde{N}(F_{(l,w)};k)+\tilde{N}(F_{(w,l)};n-k)=\binom{n}{k},
\end{equation}
\begin{equation}
    \sum\limits_{k=0}^{n}\left[\tilde{N}(F_{(l,w)};k)+\tilde{N}(F_{(w,l)};n-k)\right]=2^{n},
\end{equation}
\begin{equation}\label{ApCompl}
    \ApRel(\Nss;p)+\ApRel(\Nss^{\bot};1-p)=1.
\end{equation}
\end{remark}

\begin{remark} 
As a consequence of \eqref{ApCompl}, it follows that all the properties of the derivatives of the reliability polynomials of two dual networks proved in the previous section, Theorem \ref{kCv-Rel}, stay valid in case of the approximation polynomials. The corollaries of Theorem \ref{kCv-Rel} stay also valid, implying that the same type of high order convexity are retrieved to the approximation polynomials. Simulations on small size MMNs, showing these shape properties are in Figures \ref{fig:derivAprel_l_1} and \ref{fig:derivAprel_l}.
\end{remark}
It is proved in \cite{H2015} and \cite{L2013} that the coefficients sequence of the reliability polynomial of a MMN has the log-concavity property. It implies, as discussed above, that the coefficients functions $F_{(l,w)}$ and $F_{(w,l)}$ are concave on intervals $[l-1, n]$ and $[w-1, n]$ respectively. The concavity is preserved by the spline approximation functions $f_{(l,w)}$ and $f_{(w,l)}$. As consequence of the known information on the coefficients function \eqref{F} and \eqref{Fd}, it follows that there are two numbers $s\in \{l, l+1, ..., n-w\}$ and $t\in \{w, w+1, ..., n-l\}$ such as $N_{s}>\binom{n}{w-1}$ and $N_{t}^{\bot}>\binom{n}{l-1}$. 

\begin{property}\label{sign}
If $l>2$ and $w\geq 2$ and if $s\in \{l, l+1, ..., n-w\}$ such as $N_{s}>\binom{n}{w-1}$, then function $f_{(l,w)}$ is concave and $f_{(l,w)}(x)\geq 0$ on interval $[l-1, n-w+1]$.
\end{property}
\begin{proof}
It is either obvious or elementary to prove that $$(s-l+1)<(n-w-l+2),$$ $$(n-w-l+2)(n-w+l)>(s-l+1)(s+l-1),$$ $$s(s-l+1)<(n-w+1)(n-w-l+2).$$
Using these inequalities, one gets that $A\leq 0$, $B\geq 0$ and $C\leq 0$. The concavity of the parabola is a consequence of the negativity of $A$. The non-negativity of function $f_{(l,w)}$ is a consequence of its definition \eqref{flw} and the hypothesis on $s$.
\end{proof}

\begin{remark}
If $l\geq 2$ and $w>2$ and if $t\in \{w, w+1, ..., n-l\}$ is chosen such as $N_{t}^{\bot}>\binom{n}{l-1}$ then $$A^{\bot}\leq 0, \quad B^{\bot}\geq 0, \quad C^{\bot}\leq 0.$$
As consequence, function $f_{(w,l)}$ is concave and $f_{(w,l)}(x)\geq 0$ on $[w-1, n-l+1]$.
\end{remark}

If $x,y\in\mathbb{R}$ and $x<y$ then we denote the length of interval $[x, y]$ by $L(x;y)=y-x$ in the sequel.

\begin{property}\label{lb}
Suppose that $l>2$ and $w\geq 2$. Suppose that $s$ is chosen such that $N_{s}>\binom{n}{w-1}$. Let us denote by $V(x_{V},y_{V})$ the maximum point of function $f_{(l,w)}$ on $[l-1, n-w+1]$. Then
\begin{equation}
x_{V}\geq \frac{n-w+l}{2}.
\end{equation}
\end{property}
\begin{proof}
Let us denote the length of intervals $[l-1,n-w+1]$, $[s,n-w+1]$ and $[l-1,s]$ by $L(l-1;n-w+1)$, $L(s;n-w+1)$ and $L(l-1;s)$ respectively. Because $l \leq s \leq n-w$, it follows that $L(l-1;n-w+1) > L(s;n-w+1)$ and $L(l-1;n-w+1) > L(l-1;s)$. The abscissa of the vertex of parabola $f_{(l,w)}$ is
\begin{align*}
x_{V}&=\frac{-B}{2A}\\
&=\frac{N_{s}(n-w-l+2)(n-w+l)}{2\left[N_{s}(n-w-l+2)-\binom{n}{w-1}(s-l+1)\right]}\\
&-\frac{\binom{n}{w-1}(s-l+1)(s+l-1)}{2\left[N_{s}(n-w-l+2)-\binom{n}{w-1}(s-l+1)\right]}\\
&=\frac{n-w+l}{2}+\frac{\binom{n}{w-1}L(l-1;s)L(s;n-w+1)}{2\left[N_{s}L(l-1;n-w+1)-\binom{n}{w-1}L(l-1;s)\right]}.\\
\end{align*}
From the hypothesis $s\in [l, n-w]\cap \mathbb{N}$ and also as in the proof of Property \ref{sign}, one gets
\begin{align*}
\frac{\binom{n}{w-1}L(l-1;s)L(s;n-w+1)}{2\left[N_{s}L(l-1;n-w+1)-\binom{n}{w-1}L(l-1;s)\right]}>0,
\end{align*}
which means that the abscissa of the maximum point is greater than the middle of the interval $[l-1, n-w+1]$. The same procedure applies in case of the dual network.
\end{proof}

\begin{property}
Suppose that $l>2$ and $w\geq 2$. Suppose that $N_{n-w}>\binom{n}{w-1}$. Let us denote by $V(x_{V},y_{V})$ the maximum point of function $f_{(l,w)}$ on $[l-1, n-w+1]$. Then
\begin{equation}
 \frac{n-w+l}{2}\leq x_{V}\leq n-w+\frac{1}{2}.
\end{equation}
\end{property}
\begin{proof}
The lower bound is a particular case of Property \ref{lb}. To prove the upper bound property, we search for a real number $S$ such as
$$x_{V} \leq S + \frac{n-w+l}{2}.$$
This inequality is equivalent to
$$N_{n-w}(n-w-l+2) \leq (2S+1) \left[N_{n-w}(n-w-l+2)-\binom{n}{n-w}(n-w-l+1)\right],$$
$$0 \leq 2SN_{n-w}(n-w-l+2)-(2S+1)\binom{n}{n-w}(n-w-l+1).$$
Because of the hypothesis $N_{n-w}>\binom{n}{w-1}$, one gets
$$2S(n–w–l+2) \geq (2S+1)(n–w–l+1),$$
which meas that $2S \geq n–w–l+1$ and
$$x_{V} \leq  \frac{n-w-l+1}{2} + \frac{n-w+l}{2} = n-w+ \frac{1}{2}$$
\end{proof}
\begin{property}\label{ub}
Suppose that $l>2$ and $w\geq 2$. Suppose that $s$ is chosen such that $N_{s}>\binom{n}{w-1}$. Let us denote by $V(x_{V},y_{V})$ the maximum point of function $f_{(l,w)}$ on $[l-1, n-w+1]$. Let $E(l,w;s)=N_{s}[L(l-1;n-w+1)]^{2}-\binom{n}{w-1}L(l-1;s)[3L(l-1;n-w+1)-2L(l-1;s)].$ If \begin{equation}
    E(l,w;s)\geq 0 \label{condition}
\end{equation}
then
\begin{equation}
x_{V}\leq \frac{n-w+l}{2}+\frac{n-w-l+2}{4},
\end{equation}
\end{property}
\begin{proof}
One may write the abscissa $x_{V}$, which was computed above, as:
\begin{align*}
x_{V}&=\frac{n-w+l}{2}+\frac{n-w-l+2}{4}\\
&-\frac{N_{s}(n-w-l+2)^{2}-\binom{n}{w-1}(s-l+1)(3n-3w-l-2s+4)}{4\left[N_{s}(n-w-l+2)-\binom{n}{w-1}(s-l+1)\right]},
\end{align*}
Taking into account that, $N_{s}\geq \binom{n}{w-1}$ and $s\in [l, n-w]\cap \mathbb{N}$, one gets
\begin{align*}
&\frac{N_{s}(n-w-l+2)^{2}-\binom{n}{w-1}(s-l+1)(3n-3w-l-2s+4)}{4\left[N_{s}(n-w-l+2)-\binom{n}{w-1}(s-l+1)\right]}=\\
&\frac{N_{s}[L(l-1;n-w+1)]^{2}}{4\left[N_{s}L(l-1;n-w+1)-\binom{n}{w-1}L(l-1;s)\right]}\\
&-\frac{\binom{n}{w-1}L(l-1;s)[3L(l-1;n-w+1)-2L(l-1;s)]}{4\left[N_{s}L(l-1;n-w+1)-\binom{n}{w-1}L(l-1;s)\right]}>0.
\end{align*}
Due to \eqref{condition}, it follows the required inequality,
\begin{align*}
&x_{V}<\frac{n-w+l}{2}+\frac{n-w-l+2}{4}.
\end{align*}
\end{proof}

\begin{corollary}
Suppose that $l\geq 2$ and $w>2$. Suppose that $t$ is chosen such that $N_{t}^{\bot}>\binom{n}{l-1}$ and condition \eqref{condition} holds, i.e.
\begin{equation*}
    E(w,l;t)=N_{t}^{\bot}[L(w-1;n-l+1)]^{2}-\binom{n}{l-1}L(w-1;t)[3L(w-1;n-l+1)-2L(w-1;t)]\geq 0.
\end{equation*}
Let us denote by $V^{\bot}(x_{V^{\bot}},y_{V^{\bot}})$ the maximum point of function $f_{(w,l)}$ on $[w-1, n-l+1]$. Then the abscissa of the maximum point of the dual network has the same boundary property:
\begin{equation}
\frac{n-l+w}{2}\leq x_{V^{\bot}}\leq \frac{n-l+w}{2}+\frac{n-w-l+2}{4}.
\end{equation}
\end{corollary}
\begin{proof}
The boundary properties of the quadratic spline approximation of the coefficients function in case of the dual case are obtained following the same reasoning as in the proof of Property \ref{lb} and Property \ref{ub}, using the coefficients $A^{\bot}$ and $B^{\bot}$.
\end{proof}

\begin{remark}
Practical simulations show that the hypothesis on $s$ of being chosen such that $N_{s}>\binom{n}{w-1}$ is a necessary condition for Property \ref{sign}, Property \ref{lb} and Property \ref{ub}. The following examples show that the sufficiency does not hold. There are cases when $N_{s}<\binom{n}{w-1}$ but \eqref{condition} holds, which is shown by the following examples. The necessary and sufficient condition for Property \ref{ub} consists in both hypotheses.
\end{remark}

\begin{example}
We consider few cases of small hammock networks as in \cite{2018_CBDP}. After performing simulations taking $l\in\{3,4,5\}$ and $w\in\{3,4,5\}$ we have obtained that in cases $$(l,w)\in\{(3,4), (4,3), (3,3), (3,5), (5,3), (4,4)\}$$ inequality $N_{s}>\binom{n}{w-1}$ implies \eqref{condition}. This implication is not valid in cases $(l,w)\in\{(4,5), (5,4), (5,5)\}$. Tables \ref{tab:1}, \ref{tab:2}, and \ref{tab:3} contain the numerical results obtained in each case. The exact coefficients $N_{s}$ included in Tables \ref{tab:1}, \ref{tab:2}, and \ref{tab:3} are taken from \cite{2018_CBDP}.
\begin{table}[!ht]
\begin{center}
\begin{tabular}{|c||c|c|c|c|c||c|}
\hline
$s$&$s=5$&\cellcolor[gray]{0.8}$s=6$&$s=7$&$s\in\{8,...,14\}$&$s=15$&$s=16$\\
\hline\hline
$N_{s}$&438&\cellcolor[gray]{0.8}3072&13178&$N_{s}>\binom{20}{4}$&15468&$\binom{20}{4}=4845$\\
$E(4,5;s)$& -265128&\cellcolor[gray]{0.8}39523&1626302&$E(4,5;s)>0$&1741992&0\\
\hline
\end{tabular}
\end{center}
\caption{The case of $\Hs$ with $(l,w)=(4,5)$: if $s=6$ we have \eqref{condition} valid but the coefficient $N_6$ is less than $\binom{20}{6}.$}\label{tab:1}
\end{table}

\begin{table}[!ht]
\begin{center}
\begin{tabular}{|c||c|c|c|c|c||c|}
\hline
$s$&$s=5$&\cellcolor[gray]{0.8}$s=6$&$s=7$&$s\in\{8,...,15\}$&$s=16$&$s=17$\\
\hline\hline
$N_{s}$&36&\cellcolor[gray]{0.8}510&3334&$N_{s}>\binom{20}{3}$&4816&$\binom{20}{3}=1140$\\
$E(5,4;s)$&-36095&\cellcolor[gray]{0.8}6390&450586&$E(5,4;s)>0$&608704&0\\
\hline
\end{tabular}
\end{center}
\caption{The case of $\Hs$ of $(l,w)=(5,4)$: if $s=6$ we have \eqref{condition} valid but the coefficient $N_6$ is less than the binomial.}\label{tab:2}
\end{table}

\begin{table}[!ht]
\begin{center}
\begin{tabular}{|c||c|c|c|c|c||c|}
\hline
$s$&$s=6$&\cellcolor[gray]{0.8}$s=7$&$s=8$&$s\in\{9,...,19\}$&$s=20$&$s=21$\\
\hline\hline
$N_{s}$&994&\cellcolor[gray]{0.8}8983&50796&$N_{s}>\binom{25}{4}$&53078&$\binom{25}{4}=12650$\\
$E(5,5;s)$&-901834&\cellcolor[gray]{0.8}888337&12504244&$E(5,5;s)>0$&11493942&0\\
\hline
\end{tabular}
\end{center}
\caption{The case of $\Hs$ of $(l,w)=(5,5)$: if $s=7$ we have \eqref{condition} valid but the coefficient $N_7$ is less than the binomial.}\label{tab:3}
\end{table}

\end{example}

\begin{example}
Few interesting negative results, showing various behaviours of the approximant if $N_{s}< \binom{n}{w-1}$, are in the following cases:
\begin{itemize}
  \item If $(l,w)=(3,3)$ and $s=3$ then $N_{3}=8<\binom{9}{2}=36$ and $A=-0.25<0$, which means that the approximant is concave. As one can see, Property \ref{sign} and Property \ref{lb} are valid. But $x_{V}=22.5>5.75$, exceeding the upper bound from Property \ref{ub}.
  \item If $(l,w)=(4,4)$ and $s=4$ then $N_{4}=18<\binom{16}{3}=560$ and $A=\frac{38}{9}>0$, which means that the approximant is convex. Since $x_{V}=\frac{26}{19}<8$ neither Property \ref{sign}, nor Property \ref{lb}, nor Property \ref{ub} is valid.
  \item If $(l,w)=(4,4)$ and $s=5$ then $N_{5}=204<\binom{16}{3}=560$ and $A=-\frac{41}{9}<0$, which means that the approximant is concave, validating Property \ref{sign}. But $x_{V}=\frac{592}{41}>8$ and also $x_{V}>10.5$. It means that Property \ref{lb} is true, but Property \ref{ub} is not valid.
  \item If $(l,w)=(4,4)$-dual and $s=4$ then $N_{4}^{\bot}=24<\binom{16}{3}=560$ and $A^{\bot}=\frac{32}{9}>0$, which means that the approximant is convex. Since $x_{V^{\bot}}=\frac{1}{8}<8$ neither Property \ref{sign}, nor Property \ref{lb}, nor Property \ref{ub} is valid.
  \item If $(l,w)=(4,4)$-dual and $s=5$ then $N_{5}^{\bot}=264<\binom{16}{3}=560$ and $A^{\bot}=-\frac{76}{9}<0$, which means that the approximant is concave, validating Property \ref{sign}. But $x_{V^{\bot}}=\frac{208}{19}>8$ and also $x_{V^{\bot}}>10.5$. It means that Property \ref{lb} is true, but Property \ref{ub} is not valid.
\end{itemize}
\end{example}

\medskip

\subsection{Error estimation}\label{subsec:error-estimation}

In this subsection we also suppose that both $l\ge 2$ and $w\ge 2$ and at least one of the inequalities is strict. The small amount of initial data and of theoretic information on the reliability polynomial of a MMN arise difficulties in assessing the error of the approximation. The error of the corresponding approximated reliability polynomial is determined using the Chebychev distance between functions, because it gives information on the number of exact decimals obtained by approximation. As consequence, one can prove that the same upper bound of the error of this approximation as obtained in \cite{CD2021}. All the same, a more refined upper bound than the result from \cite{CD2021} is proved in this subsection in case of approximation using quadratic spline functions.

\begin{proposition}\label{pr:aprox-error}
Let $\Nss$ be a $(l,w)$-type MMN, and let us denote $$M=\max\{|N_k-N_{n-k}^{\bot}|\;,0\leq k\leq n\}$$
$$D=(B-B^{\bot}-2nA^{\bot})^{2}-4(A-A^{\bot})(C-C^{\bot}-B^{\bot}n-A^{\bot}n^{2}).$$
Then we have
\begin{equation}\label{eroare}
\max\{\left|\Rel(\Nss;p)-\ApRel(\Nss;p)\right|\;,0\leq p\leq1\}\leq 
\end{equation}
$$\leq \frac{(n-l-w)}{2^{n+1}}\left[M+\max\left\{\binom{n}{l-1}, \binom{n}{w-1},\left|\frac{D}{4(A-A^{\bot})}\right|\right\}\right].$$
\end{proposition}
\begin{proof} Using the definition of $\Rel$ and $\ApRel$ we obtain
\begin{equation*}
\left|\Rel(\Nss;p)-\ApRel(\Nss;p)\right|\leq
\sum_{k=0}^n\left|N_k-\tilde{N}(F_{(l,w)};k))\right|p^k(1-p)^{n-k}.
\end{equation*}

But, according to the algorithm, 
\begin{align*}
\tilde{N}(F_{(l,w)};k)&=f_{(l,w)}(k)+\frac{\Delta(k)}{2}\\
&=f_{(l,w)}(k)+\frac{1}{2}\left[\binom{n}{k}-f_{(l,w)}(k)-f_{(w,l)}(n-k)\right]\\
&=\frac{\binom{n}{k}+f_{(l,w)}(k)-f_{(w,l)}(n-k)}{2},
\end{align*}
and, using \eqref{suCo}, one can compute
\begin{align*}
N_k-\tilde{N}(F_{(l,w)};k))&=\frac{2N_k-\binom{n}{k}-f_{(l,w)}(k)+f_{(w,l)}(n-k)}{2}\\  
&=\frac{N_k-N_{n-k}^{\bot}-f_{(l,w)}(k)+f_{(w,l)}(n-k)}{2},
\end{align*}
for each $k\in\{l-1,...,n-w+1\}$. Taking into account \eqref{flw}, we evaluate as follows:
\begin{align*}
\left|\Rel(\Nss;p)-\ApRel(\Nss;p)\right|
&\leq\sum_{k=l}^{n-w}\left|\frac{N_k-N_{n-k}^{\bot}-f_{(l,w)}(k)
+f_{(w,l)}(n-k)}{2}\right|p^k(1-p)^{n-k}\\
&\leq \frac{1}{2^{n}}\left[\sum_{k=l}^{n-w}\frac{\left|N_k-N_{n-k}^{\bot}\right|}{2}+\sum_{k=l}^{n-w}\frac{\left|-f_{(l,w)}(k)+f_{(w,l)}(n-k)\right|}{2}\right]\\
&\leq \frac{(n-w-l)M}{2^{n+1}}+\frac{1}{2^{n+1}}\sum_{k=l}^{n-w}\left|f_{(l,w)}(k)-f_{(w,l)}(n-k)\right|.
\end{align*}
But, according to \eqref{flw} and \eqref{fwl}, we obtain
\begin{align*}
f_{(l,w)}(k)-f_{(w,l)}(n-k)
&=(A-A^{\bot})x^2+(B-B^{\bot}-2nA^{\bot})x+C-C^{\bot}-B^{\bot}n-A^{\bot}n^2.
\end{align*}
According to \eqref{flw} and \eqref{fwl}, the absolute maximum value of this second degree polynomial is the absolute value of the ordinate of its vertex. The maximum value of this function on $[l-1, n-w+1]$ is either in the maximum point of this polynomial or a value in the extremities of this interval. An elementary computation gives
\begin{align*}
&\max \left|f_{(l,w)}(k)-f_{(w,l)}(n-k)\right|=\max\left\{\binom{n}{l-1}, \binom{n}{w-1}, \left|-\frac{D}{4(A-A^{\bot})}\right|\right\},
\end{align*}
for all $k\in\{l, l+1, ..., n-w\}.$
Replacing this value in the previous evaluation of the approximation error one gets the required upper bound.
\end{proof}

\begin{remark}
The error of approximation can be improved by conveniently choosing the initial coefficients, $N_{s}$ and $N_{t}^{\bot}$, if possible.
\end{remark}

\section{Simulation results}\label{sec:5}

We have conducted simulations for hammock networks of size $w=3,l=5.$ The exact coefficients are extracted from \cite{2018_CBDP}. Firstly, we have implemented our algorithm and with input data $N_{l-1}=0, N_{n-w},N_{n-w+1}.$ 

\renewcommand{\arraystretch}{1.2}
\begin{table}[!ht]
    \centering
    \resizebox{\textwidth}{!}{%
    \begin{tabular}{|c|c||c|c|c|c|c|c|c|c|c|c|c|c|c|c|c|c|}
    \hline
    &     $N_i$&\cellcolor[gray]{0.8}0&\cellcolor[gray]{0.8}0&\cellcolor[gray]{0.8}0&\cellcolor[gray]{0.8}16&178&889&2562&4663&5653&4811&\cellcolor[gray]{0.8}2982&\cellcolor[gray]{0.8}1365&\cellcolor[gray]{0.8}455&\cellcolor[gray]{0.8}105&\cellcolor[gray]{0.8}15&\cellcolor[gray]{0.8}1\\
         \hline
     Alg(l-1,n-w,n-w+1)&   $\tilde{N}$ &\cellcolor[gray]{0.8}0&\cellcolor[gray]{0.8}0&\cellcolor[gray]{0.8}0&455&1365&3003&4555&5352&5256&4266&2814&\cellcolor[gray]{0.8}1365&\cellcolor[gray]{0.8}455&\cellcolor[gray]{0.8}105&\cellcolor[gray]{0.8}15&\cellcolor[gray]{0.8}1\\
         \hline
      Alg(l,n-w,n-w+1)&  $\tilde{N}$ &\cellcolor[gray]{0.8}0&\cellcolor[gray]{0.8}0&\cellcolor[gray]{0.8}0&\cellcolor[gray]{0.8}16&1330&2803&4251&5208&5244&4358&\cellcolor[gray]{0.8}2982&\cellcolor[gray]{0.8}1365&\cellcolor[gray]{0.8}455&\cellcolor[gray]{0.8}105&\cellcolor[gray]{0.8}15&\cellcolor[gray]{0.8}1\\
      \hline
         \hline
        &$N_i^{\bot}$& \cellcolor[gray]{0.8}0&\cellcolor[gray]{0.8}0&\cellcolor[gray]{0.8}0&\cellcolor[gray]{0.8}0&\cellcolor[gray]{0.8}0&\cellcolor[gray]{0.8}21&194&782&1772&2443&2114&1187&\cellcolor[gray]{0.8}439&\cellcolor[gray]{0.8}105&\cellcolor[gray]{0.8}15&\cellcolor[gray]{0.8}1\\
         \hline
         Alg(l-1,n-w,n-w+1)&$\tilde{N}$&\cellcolor[gray]{0.8}0&\cellcolor[gray]{0.8}0&\cellcolor[gray]{0.8}0&\cellcolor[gray]{0.8}0&\cellcolor[gray]{0.8}0&189&738&1179&1082&449&0&0&0&\cellcolor[gray]{0.8}105&\cellcolor[gray]{0.8}15&\cellcolor[gray]{0.8}1\\
         \hline
        Alg(l,n-w,n-w+1)&$\tilde{N}$ &\cellcolor[gray]{0.8}0&\cellcolor[gray]{0.8}0&\cellcolor[gray]{0.8}0&\cellcolor[gray]{0.8}0&\cellcolor[gray]{0.8}0&\cellcolor[gray]{0.8}21&646&1191&1227&753&200&34&\cellcolor[gray]{0.8}439&\cellcolor[gray]{0.8}105&\cellcolor[gray]{0.8}15&\cellcolor[gray]{0.8}1\\
\hline
    \end{tabular}
    }
    \caption{Exact coefficients of the 3-by-5 hammock and its dual, and approximations using our algorithm with input points $N_{l-1},N_{n-w},N_{n-w+1}$, as well with input point $N_{l},N_{n-w},N_{n-w+1}.$ }
    \label{tab:aprox3-5}
\end{table}

The results are illustrated in Table \ref{tab:aprox3-5}. The first row in each sub-table represents the exact coefficients. The second rows represent the results obtained by means of our algorithm. We notice that the first four non-zero approximations are rather big compared with the exact values, while the last approximated coefficients are much closer to the exact values. 

Secondly, we have implemented the algorithm with input data $N_{l}, N_{n-w},N_{n-w+1}.$ In this scenario we had to change the systems of equation \eqref{eq:syst21}, \eqref{eq:syst22}. More exactly, the new equations we had to solve are 

\begin{equation}
\left\{
  \begin{array}{ll}
    Al^{2}+Bl+C=N_{l} \\
    A(n-w)^{2}+B(n-w)+C=N_{n-w}\\
    A(n-w+1)^{2}+B(n-w+1)+C=\binom{n}{w-1},
  \end{array}
\right.
\end{equation}

\begin{equation}
\left\{
  \begin{array}{ll}
    A^{\perp}w^{2}+B^{\perp}w+C^{\perp}=N_{w}^{\perp} \\
    A^{\perp}(n-l)^{2}+B^{\perp}(n-l)+C^{\perp}=N_{n-l}^{\perp}\\
    A^{\perp}(n-l+1)^{2}+B^{\perp}(n-l+1)+C^{\perp}=\binom{n}{l-1}.
  \end{array}
\right.
\end{equation}

We notice from the third row in each sub-table of Table \ref{tab:aprox3-5} that using more information (extra coefficients) our algorithm outputs a much finer approximation. Indeed, when computing the error of approximation for the two versions, we have obtained $0.22$ for the version using $N_{l-1}$ and $0.18$ for the version using $N_l$, clearly pointing out the advantage of the second version. This implies that our algorithm can be adapted using more information (extra coefficients) in order to produce better results.  
\begin{figure}
\begin{subfigure}{.45\textwidth}
  \centering
  \includegraphics[width=\linewidth,height=0.8\linewidth]{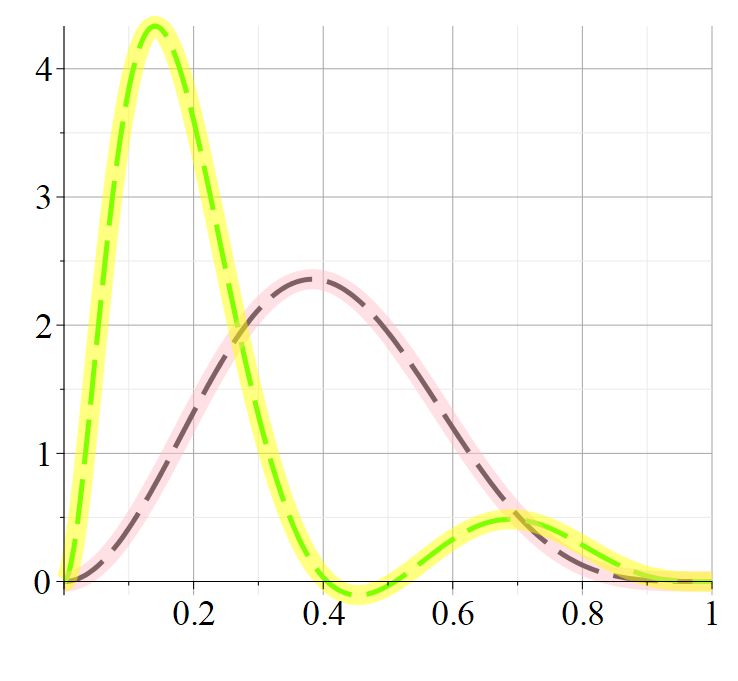}
  \caption{$k=1$}
  \label{fig:sub-first}
\end{subfigure}
\begin{subfigure}{.45\textwidth}
  \centering
  \includegraphics[width=\linewidth,height=0.8\linewidth]{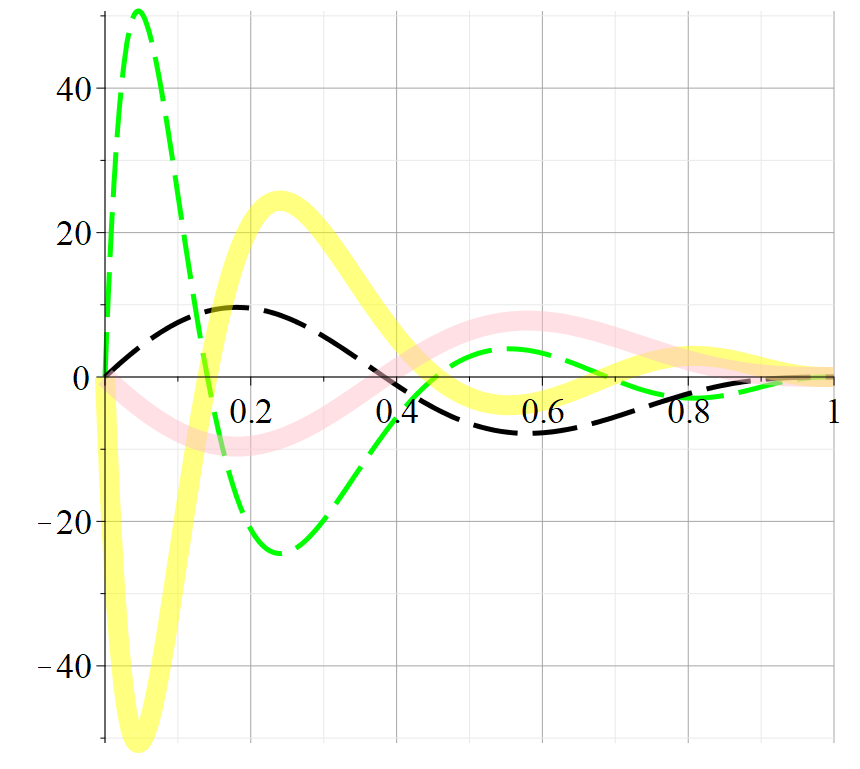}
  \caption{$k=2$}
  \label{fig:sub-second}
\end{subfigure}
\caption{$\frac{d^k}{dp^k}\Rel(\Nss;p)$ (dash black line), $\frac{d^k}{dp^k}\Rel(\Nss^{\bot};1-p)$ (solid pink line), and $\frac{d^k}{dp^k}\ApRel(\Nss;p)$ (dash green line) and $\frac{d^k}{dp^k}\ApRel(\Nss^{\bot};1-p)$ (solid yellow line) for the 3-by-5 hammock network with input data $N_{l-1},N_{n-w},N_{n-w+1}.$}
\label{fig:derivAprel_l_1}
\end{figure}

\paragraph{Properties of the approximated polynomials}
The approximated polynomials were also computed and some of the theoretical results were verified. Indeed, when we investigate the shape properties of the approximated polynomials we notice that they preserve the complementary properties. Of the properties we choose to illustrate here is the behavior of the derivatives. In Figure \ref{fig:derivAprel_l_1} we plot the derivatives of the approximated polynomials using the variant of the algorithm with input data $N_{l-1}, N_{n-w},N_{n-w+1}$, and in Figure \ref{fig:derivAprel_l} those obtained with input data $N_{l}, N_{n-w},N_{n-w+1}.$ We plot the first two derivatives, and observe the following:
\begin{itemize}
    \item $\frac{d}{dp}\ApRel(\Nss;p)=\frac{d}{dp}\ApRel(\Nss^{\bot};1-p)$;
    \item $\frac{d^2}{dp^2}\ApRel(\Nss;p)=-\frac{d^2}{dp^2}\ApRel(\Nss^{\bot};1-p)$.
\end{itemize}

\begin{figure}
\begin{subfigure}{.45\textwidth}
  \centering
  \includegraphics[width=\linewidth,height=0.8\linewidth]{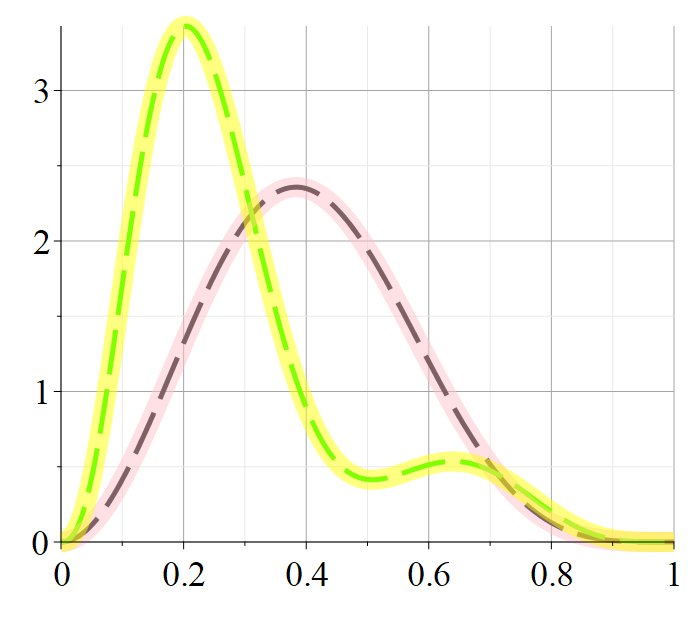}
  \caption{$k=1$}
  \label{fig:sub-third}
\end{subfigure}
\begin{subfigure}{.45\textwidth}
  \centering
  \includegraphics[width=\linewidth,height=0.8\linewidth]{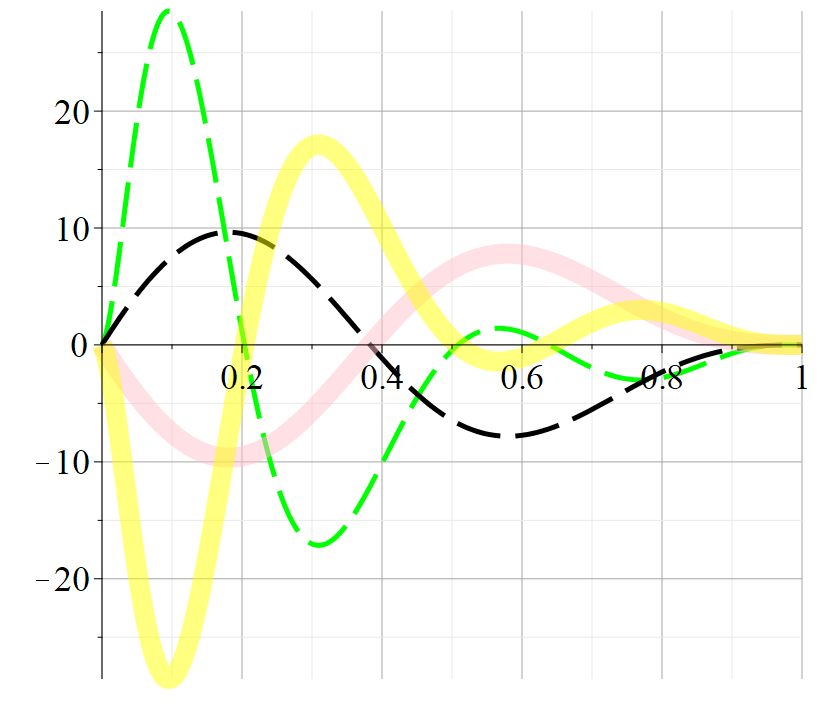}  
  \caption{$k=2$}
  \label{fig:sub-fourth}
\end{subfigure}
\caption{$\frac{d^k}{dp^k}\Rel(\Nss;p)$ (dash black line), $\frac{d^k}{dp^k}\Rel(\Nss^{\bot};1-p)$ (solid pink line), and $\frac{d^k}{dp^k}\ApRel(\Nss;p)$ (dash green line) and $\frac{d^k}{dp^k}\ApRel(\Nss^{\bot};1-p)$ (solid yellow line) for the 3-by-5 hammock network with input data $N_{l},N_{n-w},N_{n-w+1}.$}
\label{fig:derivAprel_l}
\end{figure}

\paragraph{The case of self-dual networks} As our method is using some extra information provided by means of duality, we have simulated the case of self-dual networks. We have considered the 5-by-5 hammock network, case for which we have $N_k+N_{n-k}=\binom{n}{k}.$ In Figure \ref{fig:derivApre5-5} we have plot the first derivatives of the networks and of the approximations.

\begin{figure}[!h]
\begin{subfigure}{.32\textwidth}
  \centering
  \includegraphics[width=\linewidth,height=0.7\linewidth]{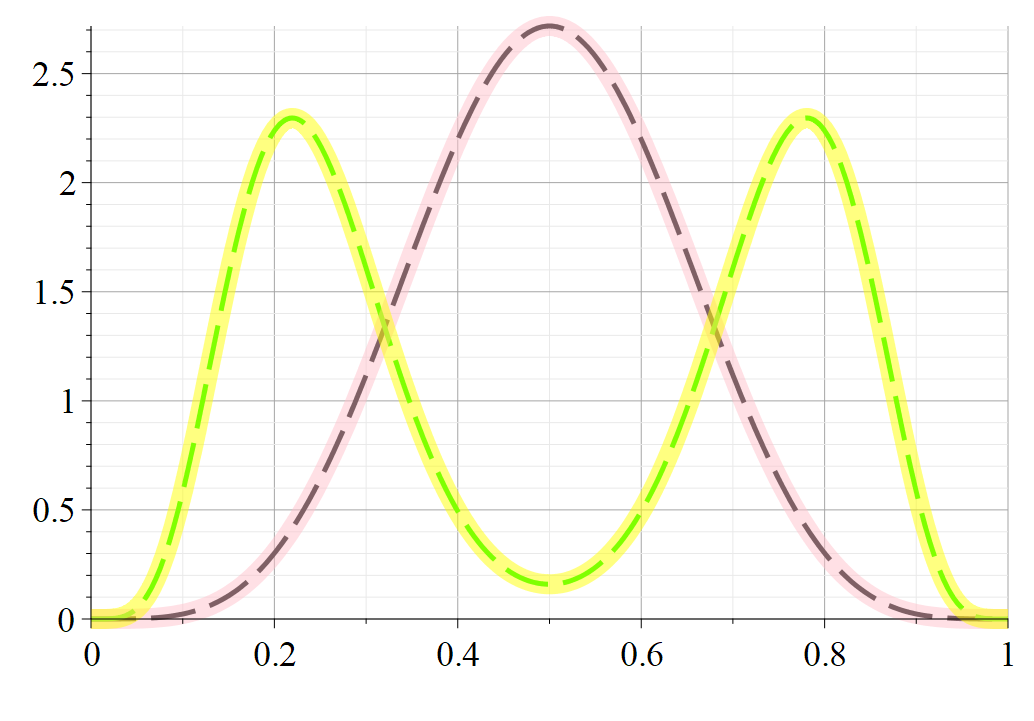}
  \caption{$k=1$}
  \label{fig:5sub-first}
\end{subfigure}
\begin{subfigure}{.32\textwidth}
  \centering
  \includegraphics[width=\linewidth,height=0.7\linewidth]{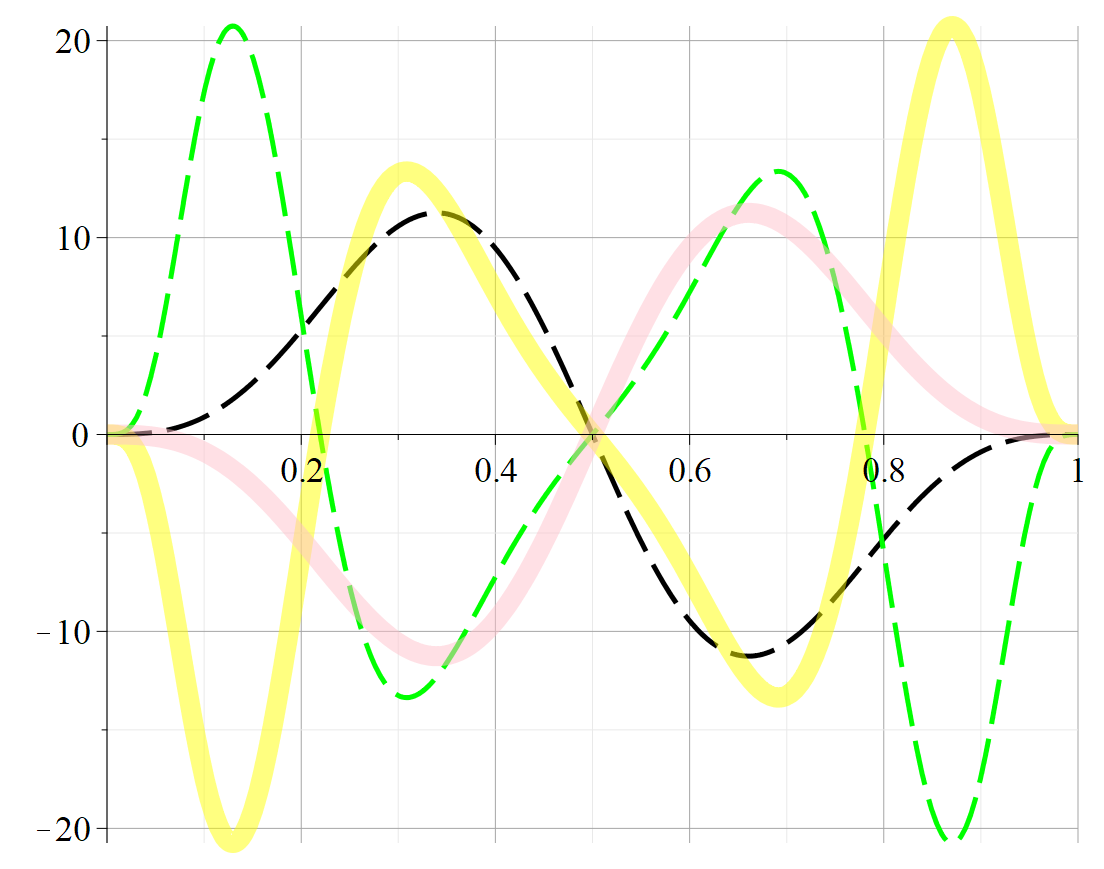}
  \caption{$k=2$}
  \label{fig:5sub-second}
\end{subfigure}
\begin{subfigure}{.32\textwidth}
  \centering
  \includegraphics[width=\linewidth,height=0.7\linewidth]{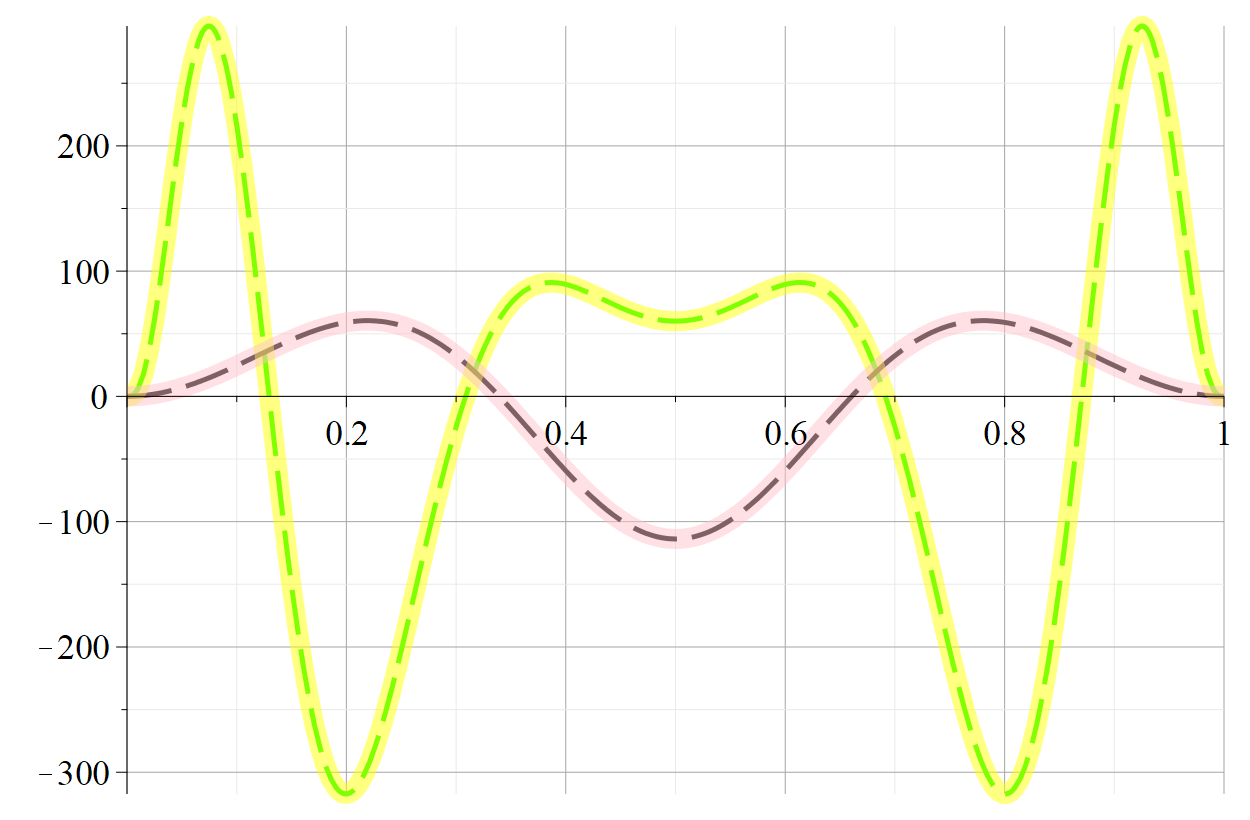}
  \caption{$k=3$}
  \label{fig:5sub-third}
\end{subfigure}
\caption{$\frac{d^k}{dp^k}\Rel(\Nss;p)$ (dash black line), $\frac{d^k}{dp^k}\Rel(\Nss^{\bot};1-p)$ (solid pink line), and $\frac{d^k}{dp^k}\ApRel(\Nss;p)$ (dash green line) and $\frac{d^k}{dp^k}\ApRel(\Nss^{\bot};1-p)$ (solid yellow line) for the 5-by-5 hammock network.}
\label{fig:derivApre5-5}
\end{figure}

Both the reliability polynomial as well as its approximant satisfy $\frac{d^{2k}}{dp^{2k}}\Rel(\Nss;0.5)=\frac{d^{2k}}{dp^{2k}}\ApRel(\Nss;0.5)=0.$ This mainly comes from the fact that $\ApRel(\Nss;0.5)=\Rel(\Nss;0.5)=0.5.$ In this case, the error of approximation is lower than $0.21.$ Notice than, even if less information is used for the particular case of self-dual networks, the error of approximation is comparable with the case of smaller networks such as the 3-by-5 hammock. This fact points out that our method could eventually be used in the case of some larger networks.  
\section{Conclusions}\label{sec:conclusion}
In this paper we have studied the shape properties of the reliability polynomials of two dual networks. Mutual behaviour of these polynomials, referring to high level convexity, inflection points, extremal properties of coefficients functions are studied. The research from \cite{CD2020} and \cite{CD2021} is developed on new bases, taking into account some requirements of shape preserving. We have proposed a technique for simultaneously approximating the reliability polynomials of two dual MMNs, choosing the input data based on their shape properties, in order to obtain results preserving some shapes. Our simulations point out that quadratic splines approach gives much better results than other types of approximation operators, if shape criteria are taken into account together with the size of the error and the complexity of computation. Some possibilities of improving the output are discussed.

\medskip
\section*{Acknowledgments}
V.-F. Dr\u{a}goi was supported by a grant of the Romanian Ministry of Education and Research, CNCS-UEFISCDI, project number PN-III-P4-ID-PCE-2020-2495, within PNCDI III -- \textit{{ThUNDER\textsuperscript{2}} = Techniques for Unconventional Nano-Designing in the Energy-Reliability Realm.}

\vspace{2cc}
\bibliographystyle{plain}
\bibliography{reliability1}

\end{document}